\documentclass[conference]{IEEEtran} 
\IEEEoverridecommandlockouts
\usepackage{setspace}
\usepackage{subcaption}
\usepackage{graphicx}
\usepackage{float}
\usepackage{amsmath}
\usepackage{amsmath,amssymb,amsfonts}
\usepackage{algorithm}
\usepackage{algorithmicx}
\usepackage{algpseudocode}
\usepackage{array}
\usepackage{amsthm}
\usepackage{amsmath}
\usepackage{amssymb}
 \usepackage{multirow}
\usepackage{mdwmath}
\usepackage{mdwtab}
\usepackage{eqparbox}
\usepackage{stfloats}
\usepackage{fixltx2e}
\usepackage{cases} 
\usepackage{caption}
\usepackage{graphicx}
\usepackage{float} 
\usepackage{booktabs} 
\usepackage{makecell} 

 \usepackage{flushend}

\usepackage{xcolor}
\usepackage{hyperref}
\hypersetup{
    colorlinks=true, 
    linkcolor=blue, 
    citecolor=blue, 
    urlcolor=blue 
}
\usepackage{url}
\usepackage{textcomp}
\usepackage{cite}
\def\BibTeX{{\rm B\kern-.05em{\sc i\kern-.025em b}\kern-.08em
    T\kern-.1667em\lower.7ex\hbox{E}\kern-.125emX}}
\allowdisplaybreaks[4]

\newtheorem{lemma}{\bf Lemma}

\newtheorem{remark}{Remark}

\newtheorem{conjecture}{\bf Conjecture}
\usepackage{wrapfig}

\newcommand{\algoref}[1]{Algorithm~\ref{#1}}

\makeatletter
\renewcommand{\maketag@@@}[1]{\hbox{\m@th\normalsize\normalfont#1}}%
\makeatother
\UseRawInputEncoding 

\makeatletter
\newcommand{\linebreakand}{%
  \end{@IEEEauthorhalign}
  \hfill\mbox{}\par
  \mbox{}\hfill\begin{@IEEEauthorhalign}
}
\makeatother

\begin{document}

\title
{Towards a Theoretical Framework for Robust Node Deployment in Cooperative ISAC Networks}
\author{\IEEEauthorblockN{ Haojin Li\textsuperscript{12}, Kaiqian Qu\textsuperscript{3}, Chen Sun\textsuperscript{2*}, Anbang Zhang\textsuperscript{4}, Xiaoxue Wang\textsuperscript{2}, Wenqi Zhang\textsuperscript{2},  Haijun Zhang\textsuperscript{1}}
\IEEEauthorblockA{
\textsuperscript{1} Beijing Engineering and Technology Research Center for Convergence Networks and Ubiquitous Services, \\University of Science and Technology Beijing, Beijing, China\\
\textsuperscript{2} Wireless Network Research Department, Sony China Research Laboratory, Beijing, China\\
\textsuperscript{3} School of Information Science and Engineering, Southeast University, Nanjing 210096, China\\
\textsuperscript{4} School of Control Science and Engineering, {Shandong University},
Jinan, China\\
{\textsuperscript{*}Corresponding Author:  Chen Sun (Email: chen.sun@sony.com).}
}
}



\maketitle


\begin{abstract}
This paper investigates node deployment strategies for robust multi-node cooperative localization in integrated sensing and communication (ISAC) networks.We first analyze how steering vector correlation across different positions affects localization performance and introduce a novel distance-weighted correlation metric to characterize this effect. Building upon this insight, we propose a deployment optimization framework that minimizes the maximum weighted steering vector correlation by optimizing simultaneously node positions and array orientations, thereby enhancing worst-case network robustness. Then, a genetic algorithm (GA) is developed to solve this min–max optimization, yielding optimized node positions and array orientations. Extensive simulations using both multiple signal classification (MUSIC) and neural-network (NN)-based localization validate the effectiveness of the proposed methods, demonstrating significant improvements in robust localization performance.
\end{abstract}


\section{Introduction}
\IEEEPARstart{I}{ntegrated} sensing and communication (ISAC) has emerged as a key paradigm for sixth-generation (6G) networks, enabling the joint utilization of spectrum and hardware resources\cite{1}. While existing research has largely concentrated on single-node ISAC systems, covering aspects such as signal design and beamforming \cite{2,3,4}, these architectures are inherently limited in terms of coverage, angular diversity, and robustness against interference. Multi-node cooperative ISAC systems overcome these limitations by leveraging spatially distributed nodes, thereby extending coverage, offering diverse sensing perspectives, and enhancing interference resilience \cite{10972176}.

Despite these advantages, the deployment of cooperative nodes plays a decisive role in determining system performance, particularly for localization robustness and continuity within the service region. Most prior works have emphasized collaborative frameworks\cite{10972176,Rahman2020,9,Meng2025}, collaborative MIMO techniques \cite{10465053,6,Meng2024,8}, collaborative signal processing\cite{10906066,Liu2024TargetLW,Wei2024}, whereas the fundamental impact of node deployment on localization performance remains underexplored.
Authors in \cite{10.1109/TWC.2024.3491356} analyzed the communication and sensing functions in multi-node cooperation and adopted the Cram\'{e}r Rao lower bound (CRLB) to evaluate localization accuracy. Nevertheless, \cite{10.1109/TWC.2024.3491356} assume that node locations follow a two-dimensional Poisson point process (PPP), implying random distributions. Consequently, their conclusions are essentially statistical in nature and provide little guidance for practical deployment design. 

To address these challenges, this paper investigates node deployment strategies for robust cooperative localization in ISAC networks. We focus on guaranteeing worst-case localization accuracy so as to ensure spatially continuous and resilient positioning services. To this end, we propose a novel distance-weighted steering-vector correlation metric that captures the impact of node geometry on localization accuracy. Based on this metric, we formulate a min–max optimization problem and develop a genetic algorithm (GA) to jointly optimize node positions and array orientations. The proposed framework is validated using two representative localization methods: the classical multiple signal classification (MUSIC) algorithm and neural-network (NN)-based localization. Extensive simulations demonstrate that our deployment optimization framework significantly enhances robustness compared to conventional fixed or random deployment strategies.

\section{System model}

\begin{figure}[!t]
    \centering
    \includegraphics[width=0.9\linewidth]{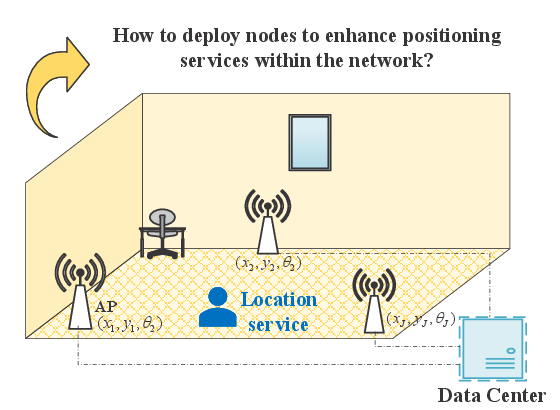}
    \caption{Illustration of the multi-node cooperative localization system.}
    \label{fig:system_model}
\end{figure}

We consider an indoor multi-node cooperative localization scenario, as illustrated in Fig.~\ref{fig:system_model}. 
The network consists of $J$ APs, each equipped with a uniform linear array (ULA) comprising $N$ antennas with spacing $d$ and carrier wavelength $\lambda$. 
The deployment parameters of the $j$-th AP are denoted as $(x_j, y_j, \theta_j)$, where $(x_j, y_j)$ represents the Cartesian coordinate of the AP and $\theta_j$ denotes the orientation angle of its ULA relative to the $x$-axis. 
A location service is provided for a target device located within the coverage area, where all APs simultaneously collect received signals and forward them to a centralized data center for cooperative processing. 
This distributed antenna architecture enables joint estimation of the target’s position by exploiting the spatial information embedded in the received signals.

\subsection{Signal Model}
The received signals are forwarded to a centralized data center, where signal-level fusion is performed to exploit spatial diversity across nodes. After fusion, the aggregated received signal can be expressed as
\begin{equation}\label{eq1}
\mathbf{y}(t) = \mathbf{a}(\mathbf{u},  \mathbf{p})\, s(t) + \mathbf{n}(t) \in \mathbb{C}^{NJ},
\end{equation}
where $\mathbf{a}(\mathbf{p}) \in \mathbb{C}^{NJ}$ denotes the composite steering vector determined by the target position $\mathbf{p}$ and the antanna position $\mathbf{u}$, $s(t)$ is the transmitted signal, and $\mathbf{n}(t)$ represents the additive noise vector.\footnote{The signal $s(t)$ may be an active positioning waveform transmitted by the target if it is equipped with communication capability, or an echo of a sensing signal reflected at the target.}
\subsubsection{Steering Vector Model}
The steering vector $\mathbf{a}(\mathbf{u},\mathbf{p})$ in \eqref{eq1} captures the geometric propagation relationship between the target position $\mathbf{p}$ and the antenna array positions $\mathbf{u}$. Specifically, the composite steering vector can be expressed as
\begin{equation}
\mathbf{a}(\mathbf{u},\mathbf{p}) =
\begin{bmatrix}
\mathbf{a}_1(\mathbf{u}_1,\mathbf{p}) ,
\ldots ,
\mathbf{a}_J(\mathbf{u}_J,\mathbf{p})
\end{bmatrix}^T\in \mathbb{C}^{NJ},
\end{equation}
where $\mathbf{a}_j(\mathbf{u}_j,\mathbf{p})\in\mathbb{C}^{N}$ denotes the local steering vector of the $j$-th node, and is given element-wise by
\begin{equation}
a_{j,n}(\mathbf{u}_{j,n},\mathbf{p})
= \sqrt{\frac{1}{NJ}}\exp\!\left(-\jmath \tfrac{2\pi}{\lambda} \left\|\mathbf{p}-\mathbf{u}_{j,n}\right\|_2\right),
\end{equation}
with $\mathbf{u}_{j,n}$ denoting the physical position of the $n$-th antenna at node $j$ and $\lambda$ the carrier wavelength.

For notational compactness and deployment optimization, each node array can be parameterized by its center coordinates and orientation angle. Let $\mathbf{v}_j \triangleq [x_j,\,y_j,\,\theta_j]^{\top}$ denote the deployment parameters of the $j$-th node, with $(x_j,y_j)$ being the array center and $\theta_j$ the orientation. The $n$-th antenna position is then
\begin{equation}
\mathbf{u}_{j,n} = 
\begin{bmatrix}x_j\\y_j\end{bmatrix}
+ \left(n-\tfrac{N+1}{2}\right)d\,
\begin{bmatrix}\cos\theta_j \\ \sin\theta_j\end{bmatrix},
\end{equation}
where $d$ is the inter-element spacing. Substituting this expression into the element-wise model yields the compact form $\mathbf{a}(\mathbf{v},\mathbf{p})$, where $\mathbf{v}=[\mathbf{v}_1^{\top},\ldots,\mathbf{v}_J^{\top}]^{\top}$ collects the deployment parameters of all nodes.
\subsubsection{Sample Covariance Matrix Estimation}
The covariance matrix is estimated from $T$ snapshots of the received signal vector $\mathbf{y}(t)$ as
\begin{equation}
\widehat{\mathbf{R}} = \frac{1}{T}\sum_{t=1}^{T}\mathbf{y}(t)\mathbf{y}^{H}(t) \in \mathbb{C}^{NJ \times NJ}.
\end{equation}
The sample covariance $\widehat{\mathbf{R}}$ encapsulates the spatial correlation among antenna elements across different nodes, and serves as the key input for both classical subspace-based estimators and data-driven localization models.

\section{Node Deployment Optimization Based on Steering Vector Correlation}

In this section, we investigate the optimization of node deployment for robust cooperative localization.
\subsection{Robustness Analysis}

To elucidate the fundamental impact of node deployment on localization robustness, this section first introduces two lemmas. These lemmas, derived respectively from neural-network–based localization methods and spatial-spectrum–based approaches, analyze the role of steering vector correlation in the formation of localization errors.
\begin{lemma}
\label{lemma:correlation}
In NN-based grid localization, increased inner-product similarity among steering vectors across different positions reduces feature discriminability, thereby impairing classification accuracy.
\end{lemma}
\begin{proof}
Let \(\mathbf{y}(t)\in\mathbb{C}^{N}\) denote the received snapshot \emph{column} vector. 
The network input feature is the vectorized sample covariance
\begin{equation}
    \widehat{\mathbf{r}} \;=\; \mathrm{vec}\!\big(\widehat{\mathbf{R}}\big), ~\widehat{\mathbf{R}}= \frac{1}{T}\sum_{t=1}^{T}\mathbf{y}(t)\mathbf{y}^{H}(t)
\approx E\,\mathbf{a}_i\mathbf{a}_i^{H}+\sigma^{2}\mathbf{I},
\end{equation}
where \(E=\mathbb{E}[|s(t)|^{2}]\) is the signal power and \(\mathrm{vec}(\cdot)\) stacks matrix columns.

The neural network is trained to classify among \(K\) classes based on covariance features. 
We define the inter-class separability between classes \(i\) and \(j\) as the Frobenius norm of the difference between their covariance matrices:
\begin{equation}
    \Delta \mathbf{R}_{ij} \;\triangleq\; \mathbf{R}_i - \mathbf{R}_j 
    \;=\;E\!\left(\mathbf{a}_i \mathbf{a}_i^{H} - \mathbf{a}_j \mathbf{a}_j^{H}\right).
\end{equation}

Assuming normalized steering vectors $\|\mathbf{a}_i\| = \|\mathbf{a}_j\| = 1$, we have:
\begin{equation}
    \begin{split}
            \|\Delta \mathbf{R}_{ij}\|_F &= E \left\| \mathbf{a}_i \mathbf{a}_i^H - \mathbf{a}_j \mathbf{a}_j^H \right\|_F \\
    &= E\sqrt{2\left(1 - |\mathbf{a}_i^H \mathbf{a}_j|^2 \right)}.
    \end{split}
\end{equation}

This result shows that the inter-class distance is inversely related to the inner-product similarity between steering vectors. As $|\mathbf{a}_i^H \mathbf{a}_j| \to 1$, we have $\|\Delta \mathbf{R}_{ij}\|_F \to 0$, which implies higher classification ambiguity.
\end{proof}

\begin{lemma}
\label{lemma:correlation2}
In spatial spectrum estimation methods, a high correlation between the steering vectors of the true and a false target location increases the risk of mislocalization.
\end{lemma}

\begin{proof}
Let $\mathbf{a} = \mathbf{a}(\mathbf{p}_{\mathrm{true}})$ and $\mathbf{b} = \mathbf{a}(\mathbf{p}_{\mathrm{false}})$ be the steering vectors for the true and a false location, respectively. Spatial spectrum methods (e.g., MUSIC, MVDR, Beamforming) compute a power score $g(\mathbf{p}) = \mathbf{a}^H(\mathbf{p}) \widehat{\mathbf{T}} \mathbf{a}(\mathbf{p})$ for each location, where $\widehat{\mathbf{T}}$ is a positive semi-definite matrix (e.g., the noise subspace projector, inverse covariance or covariance).

Decompose the false steering vector $\mathbf{b}$ relative to the true one $\mathbf{a}$:
\begin{equation}
\mathbf{b} = \alpha \mathbf{a} + \mathbf{r}, 
\end{equation}
where$ \alpha = \mathbf{a}^H\mathbf{b},\ \mathbf{a}^H\mathbf{r}=0,\ \|\mathbf{r}\|^2 = 1 - |\alpha|^2$.
Substituting into the spectrum expression yields:
\begin{align}
g(\mathbf{b}) &= |\alpha|^2 \mathbf{a}^H \widehat{\mathbf{T}} \mathbf{a} + 2\Re\left\{ \alpha^* \mathbf{a}^H \widehat{\mathbf{T}} \mathbf{r} \right\} + \mathbf{r}^H \widehat{\mathbf{T}} \mathbf{r}.
\end{align}

As $|\alpha| = |\mathbf{a}^H\mathbf{b}| \to 1$ (high correlation):
\begin{itemize}
  \item The first term approaches $g(\mathbf{a})$.
  \item The second cross-term is bounded by $2\|\widehat{\mathbf{T}}\|_2 |\alpha|\sqrt{1-|\alpha|^2} \to 0$.
  \item The third term is bounded by $\|\widehat{\mathbf{T}}\|_2 (1-|\alpha|^2) \to 0$.
\end{itemize}
Therefore, $g(\mathbf{b}) \approx g(\mathbf{a})$, making the false location spectrally indistinguishable from the true one and greatly increasing the likelihood of mislocalization.
\end{proof}

Therefore, the correlation among steering vectors directly determines the probability of correctly identifying the true spatial position. However, the impact of misidentification on localization error also depends on the distance between the correct and erroneous spatial positions. Building upon this observation, we put forward a reasonable conjecture: 

\begin{conjecture}
The localization error is primarily governed by the distance-weighted correlation of steering vectors within the network, i.e.,
\begin{equation}
    \mathrm{RMSE}_{\max} \ \propto  \max_{i,j} \ \rho(\mathbf{p}_i,\mathbf{p}_j,\mathbf{v}),
\end{equation}
where $\rho(\mathbf{p}_i,\mathbf{p}_j,\mathbf{v})=\big|\mathbf{a}_i^H \mathbf{a}_j \big| \cdot d_{ij}^{\alpha}$, $\mathbf{a}_i$ and $\mathbf{a}_j$ denote the steering vectors associated with candidate positions $i$ and $j$, $d_{ij}$ represents their spatial distance, and $\alpha$ is a weighting factor.
\end{conjecture}

To validate the proposed conjecture, we conduct Monte Carlo experiments with 1000 random deployment realizations. As shown in Fig.~\ref{fig:correlation_experiment},  RMSE exhibits a clear positive correlation with the maximum distance-weighted steering vector correlation across different values of $\alpha$.

\begin{figure}[!t]
    \centering
    \includegraphics[width=0.85\linewidth]{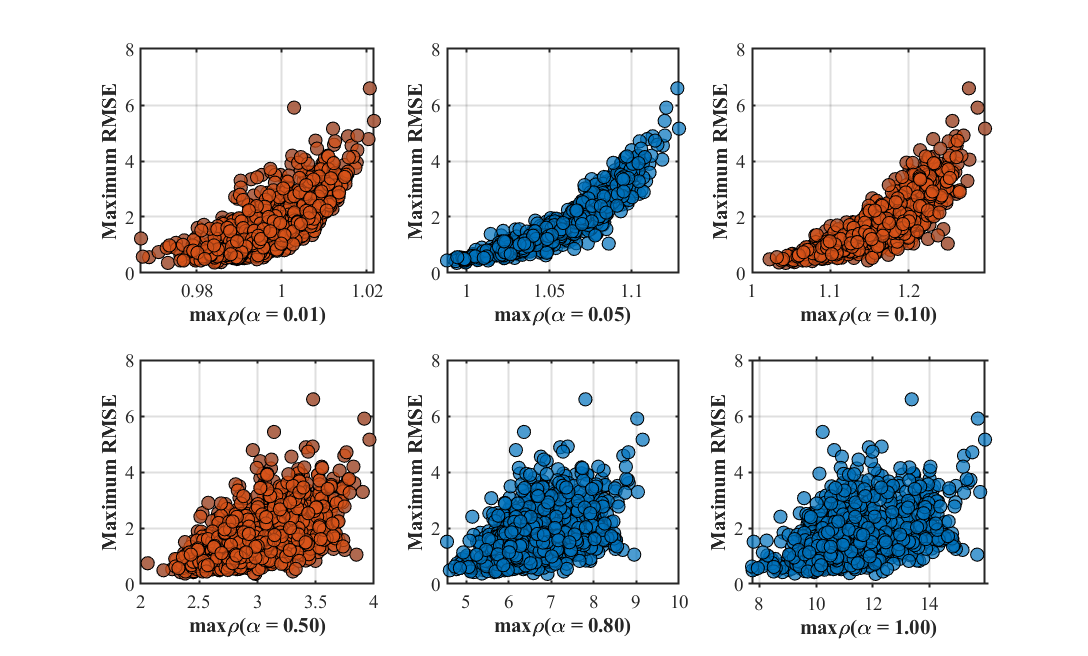}
    \caption{Distance-weighted correlation vs Maximum RMSE under 1000 random deployment realizations.}
    \label{fig:correlation_experiment}
\end{figure}

\begin{figure}[!t]
    \centering
    \includegraphics[width=0.7\linewidth]{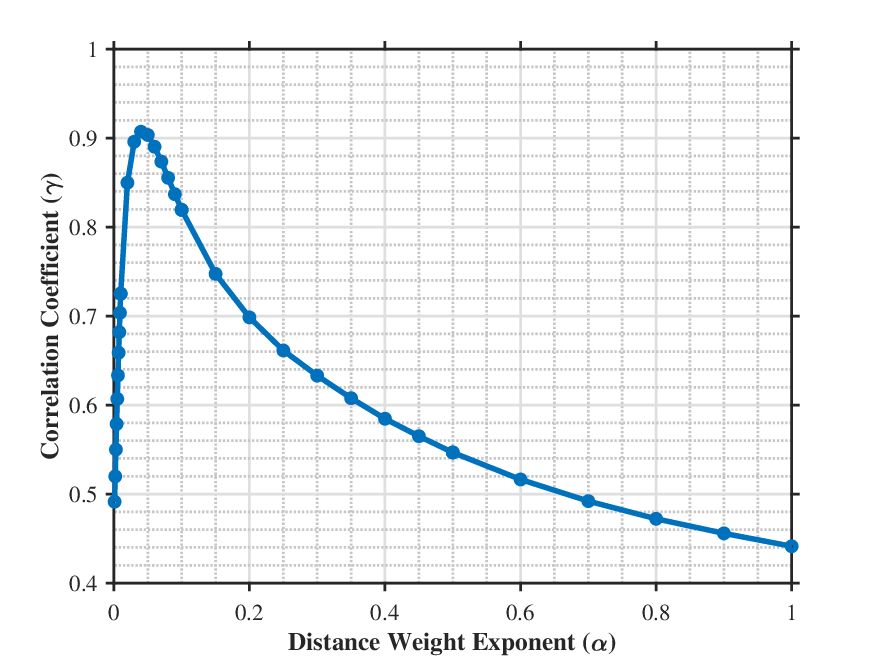}
    \caption{Correlation Analysis for Different Distance Weight Exponents}
    \label{fig:correlation}
\end{figure}

We evaluate the correlation between the maximum localization error and the maximum distance-weighted steering-vector correlation. Specifically, we compute the Pearson correlation coefficient
\begin{equation}
\gamma(\alpha) 
= \mathrm{corr}\!\left(\mathrm{RMSE}_{\max}^{(k)}, \, \max_k \rho(\alpha) \right),
\end{equation}
where $\mathrm{RMSE}_{\max}^{(k)}$ denotes the worst-case localization error across all target positions in the network under the $k$-th random development, and $\max_k \rho(\alpha)$ represents the maximum weighted steering-vector correlation with weighting factor $\alpha$. 

The simulation results in Fig.~\ref{fig:correlation} demonstrate that $\gamma(\alpha)$ remains consistently high for different values of $\alpha$, thereby confirming the strong dependence between the worst-case localization error and the maximum weighted steering-vector correlation. In particular, the correlation achieves its peak around $\alpha = 0.05$, indicating that this weighting factor best captures the structural characteristics governing localization robustness. The consistent trend across a large number of trials provides strong empirical evidence in support of our conjecture.

\begin{remark}
Since our optimization framework aims to guarantee the robustness of network-wide localization performance, the node deployment and selection strategies can be regarded as offline optimization procedures. In this context, it is reasonable to determine the distance-weight exponent $\alpha$ by Monte Carlo simulations, so as to capture the statistical relationship between steering vector correlation and localization robustness in the considered network.
\end{remark}


\subsection{Problem Formulation}

The analysis confirms that localization robustness is governed by the correlation between steering vectors at different locations. Our conjecture and empirical results show that the worst-case error is strongly tied to the maximum \emph{distance-weighted} steering vector correlation. We therefore formulate the robust deployment problem as minimizing this maximum correlation.

Let $\mathcal{P}$ be the set of all potential target positions. For a deployment configuration $\mathbf{v}$ (encoding node positions and orientations), the distance-weighted correlation between any two points $\mathbf{p}_i, \mathbf{p}_j \in \mathcal{P}$ is defined as:
\begin{equation}
\rho(\mathbf{v}, \mathbf{p}_i, \mathbf{p}_j) = \left| \tilde{\mathbf{a}}(\mathbf{v}, \mathbf{p}_i)^H \tilde{\mathbf{a}}(\mathbf{v}, \mathbf{p}_j) \right| \cdot \| \mathbf{p}_i - \mathbf{p}_j \|_2^{\alpha}.
\end{equation}
The optimization goal is to find the deployment $\mathbf{v}$ that minimizes the maximum correlation over all pairs of distinct points:
\begin{equation}
\min_{\mathbf{v}} \; \max_{\substack{\mathbf{p}_i, \mathbf{p}_j \in \mathcal{P} \\ i \neq j}} 
\; \rho(\mathbf{v}, \mathbf{p}_i, \mathbf{p}_j ).
\end{equation}
This min-max formulation directly targets the worst-case ambiguity in the network, ensuring uniformly robust localization performance. We propose a genetic algorithm to solve this non-convex optimization problem efficiently.
\subsection{GA-based Node Deployment Optimization}
\addtolength{\topmargin}{0.051in}
\subsubsection{Encoding Scheme}
Each chromosome employs a real-coded encoding scheme to represent one deployment configuration:
\[
\mathbf{c}=\big[\mathbf{v}_1^\top,\mathbf{v}_2^\top,\ldots,\mathbf{v}_J^\top\big]^\top\in\mathbb{R}^{3J},
\]
with bounds
\((x_j,y_j)\in\mathcal{R}\ (\text{feasible region}),\theta_j\in[0,2\pi),\quad j=1,\ldots,J.\)

\subsubsection{Fitness Function}
For a chromosome $\mathbf{c}$, construct all antenna positions via \(\mathbf{u}_{j,n}\), compute steering vectors on all grid points, and evaluate the maximum weighted correlation:
\begin{equation}
f(\mathbf{c}) \;=\; 
\max_{1\le i<j\le|\mathcal{P}|}\;
\big|\tilde{\mathbf{a}}(\mathbf{c},\mathbf{p}_i)^H\tilde{\mathbf{a}}(\mathbf{c},\mathbf{p}_j)\big|\,
\|\mathbf{p}_i-\mathbf{p}_j\|_2^{\alpha}
\;。
\label{eq:fitness}
\end{equation}

\subsubsection{Genetic Operators}

\paragraph*{\textbf{Initialization}}
Randomly sample $(x_j,y_j)$ uniformly in $\mathcal{R}$ (or from a given candidate set) and $\theta_j\sim\mathcal{U}[0,2\pi)$ to form the initial population $\{\mathbf{c}^{(0)}_m\}_{m=1}^{P}$.

\paragraph*{\textbf{Selection (Tournament)}}
Pick $k$ chromosomes uniformly at random and select the one with the smallest $f(\cdot)$ as parent. Repeat to obtain the parent pool.

\paragraph*{\textbf{Crossover (SBX, real-coded)}}
 For two parents $\mathbf{c}^{(a)}$ and $\mathbf{c}^{(b)}$, generate two children $\mathbf{z}^{(1)}$ and $\mathbf{z}^{(2)}$ element-wise via  (SBX) with distribution index $\eta_c$ and crossover probability $p_c$:
\begin{equation}\label{eq19}
\begin{split}
        z^{(1)}_q = \tfrac{1}{2}\big[(1+\beta_q)c^{(a)}_q + (1-\beta_q)c^{(b)}_q\big],\\
 z^{(2)}_q = \tfrac{1}{2}\big[(1-\beta_q)c^{(a)}_q + (1+\beta_q)c^{(b)}_q\big],
\end{split}
\end{equation}
where 
\[\beta_q=\begin{cases}
(2u)^{1/(\eta_c+1)}, & u\le 1/2,\\[2pt]
\big(1/(2-2u)\big)^{1/(\eta_c+1)}, & u>1/2,
\end{cases}\\\]
with $u\sim\mathcal{U}(0,1)$ and $q=1,\ldots,3J$. Project $(x_j,y_j)$ back to $\mathcal{R}$ and wrap $\theta_j$ to $[0,2\pi)$.

\renewcommand{\algorithmicrequire}{\textbf{Input:}}
	\renewcommand{\algorithmicensure}{\textbf{Output:}}
\begin{algorithm}[t]
\caption{GA-based Node Deployment Optimization}
\label{alg:GA-deploy}
\begin{algorithmic}[1]
\Require  $P$, $p_c$, $p_m$, $\eta_c, \eta_m$, $G_{\max}$, $N_{e}$, $\mathcal{R}$, $\mathcal{P}$.
\Ensure Best deployment  $\mathbf{v}=\mathbf{c}^*$.

\State \textbf{Initialization:} Randomly generate $P$ chromosomes 
$\{\mathbf{c}_m^{(0)}\}_{m=1}^P$, and evaluate fitness $f(\mathbf{c}_m^{(0)})$. Set $\mathbf{c}^*=\arg\min f(\mathbf{c}_m^{(0)})$.

\For{$g=0$ to $G_{\max}-1$}
  \State  Apply tournament selection to form parent pool.
  
  \State  \parbox[t]{\linewidth}{\raggedright For each parent pair,
  generate offspring via Eq. (\ref{eq19}). Project offspring back into feasible range: 
  $z_q \leftarrow \min(\max(z_q,L_q),U_q)$.}
  
  \State  \parbox[t]{\linewidth}{\raggedright Perform the mutation operation for each gene via Eq. (\ref{eq20}), and project as
  $z_q' \leftarrow \min(\max(z_q',L_q),U_q)$. } 

  \State \parbox[t]{\linewidth}{\raggedright Carry over best $N_e$ parent chromosomes $\{\mathbf{c}_e^{(g)}\}$ to \\next generation $\{\mathbf{c}_m^{(g+1)}\}=\{\mathbf{c}_e^{(g)},\mathbf{z}\}$.}
  
  \State  Evaluate $f(\mathbf{c}_m^{(g+1)})$.

  \State Update $\mathbf{c}^* \leftarrow 
  \arg\min\{ f(\mathbf{c}^*), \min_m f(\mathbf{c}_m^{(g+1)})\}$.
\EndFor
\State \Return $\mathbf{v}=\mathbf{c}^*$
\end{algorithmic}
\end{algorithm}
\paragraph*{\textbf{ Mutation (Polynomial)}}
Adopt polynomial mutation for each gene with probability $p_m$. 
Let a gene $z_q$ be bounded by $[L_q,U_q]$, and define
\[
\delta_1=\frac{z_q-L_q}{U_q-L_q},\quad 
\delta_2=\frac{U_q-z_q}{U_q-L_q},\quad 
m=\frac{1}{\eta_m+1}.
\]
The mutation step $\Delta$ is drawn as
\[
\Delta=
\begin{cases}
\big(2u + (1-2u)(1-\delta_1)^{\eta_m+1}\big)^{m}-1, & u\le \tfrac{1}{2},\\[4pt]
1-\big(2(1-u) + 2(u-\tfrac{1}{2})(1-\delta_2)^{\eta_m+1}\big)^{m}, & u>\tfrac{1}{2},
\end{cases}
\]
and the mutated gene is
\begin{equation}\label{eq20}
  z_q' \;=\; z_q \;+\; \Delta\,(U_q-L_q).  
\end{equation}

Finally, project $z_q'$ back to $[L_q,U_q]$ if necessary.
The distribution index $\eta_m$ controls locality (larger $\eta_m$ $\Rightarrow$ smaller perturbations).

\paragraph*{\textbf{Elitism \& Stopping}}
Carry over the best $N_e$ parent chromosomes $\{\mathbf{c}_e^{(g)}\}$ to the next generation ($P=N_e+k$). Stop when the maximum generation $G_{\max}$ is reached.

\begin{figure}[!t]
    \centering
    \includegraphics[width=0.7\linewidth]{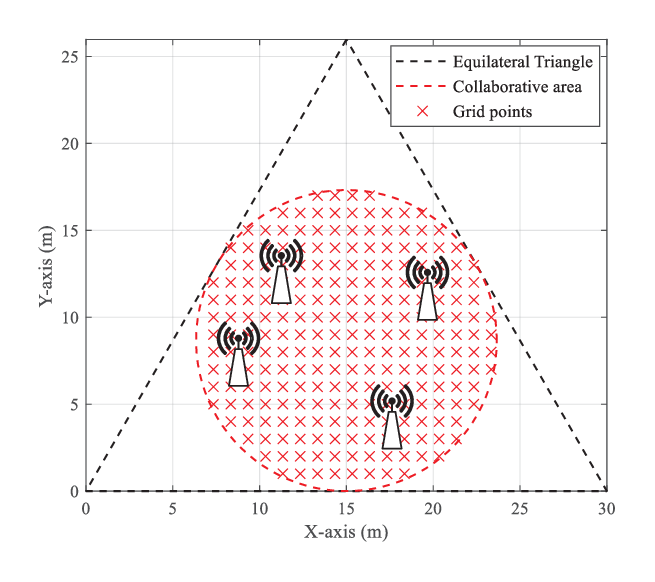}
    \caption{Simulation scenario of multi-node cooperative localization.}
    \label{result1}
\end{figure}

The details of the proposed GA-based approach for solving
 problem (17) are summarized in \algoref{alg:GA-deploy}.

\section{SIMULATIONS AND DISCUSSION}
\label{va}
This section validates the proposed node deployment scheme. We consider an indoor multi-node cooperative ISAC system with a circular coverage area (inscribed in a 30m equilateral triangle, see Fig. \ref{result1}). Targets are localized on a 1m-resolution grid $\mathcal{P}$.  Unless otherwise specified, the carrier frequency is $2.4$\,GHz, SNR $=\frac{\mathbb{E}(s^2(t))}{\mathbb{E}(n^2(t))}$ is set to $0$ dB, the number of snapshots is $200$. $J$ nodes will be deployed within the incircle, each equipped with a ULA of $N=4$ antennas.

For the GA parameters, we set the population size $P=100$, crossover probability $p_c=0.8$, mutation probability $p_m=0.2$, number of elites $N_e=4$, and maximum generations $G_{\max}=500$.

\begin{figure*}[!t]
    \centering
        \begin{minipage}[b]{0.32\linewidth}
        \centering
      \includegraphics[width=\linewidth]{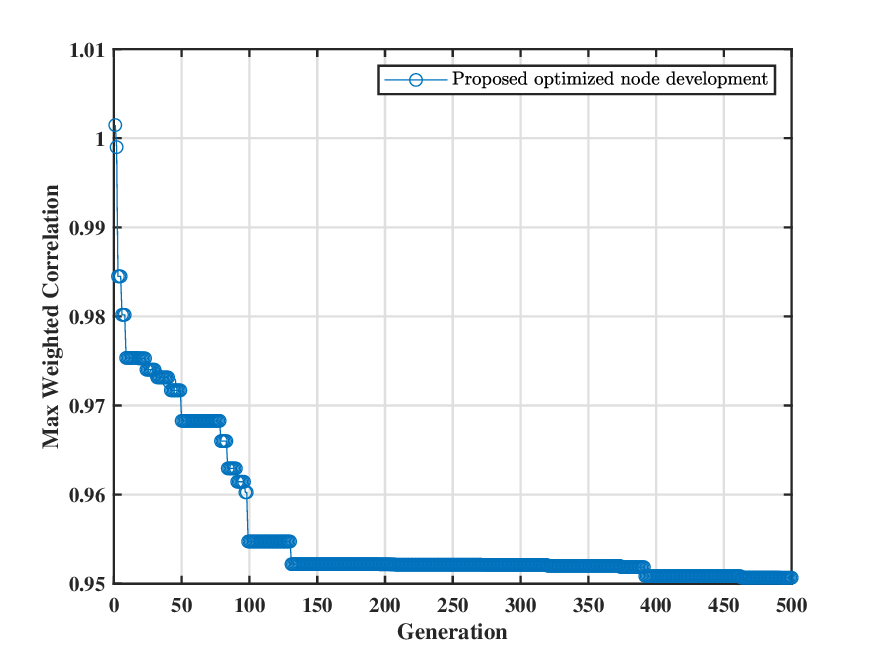}
    \caption{Convergence behavior of the proposed GA-based optimization schemes.}
    \label{r1}
    \end{minipage}
    \hfill
    \begin{minipage}[b]{0.32\linewidth}
        \centering
        \includegraphics[width=\linewidth]{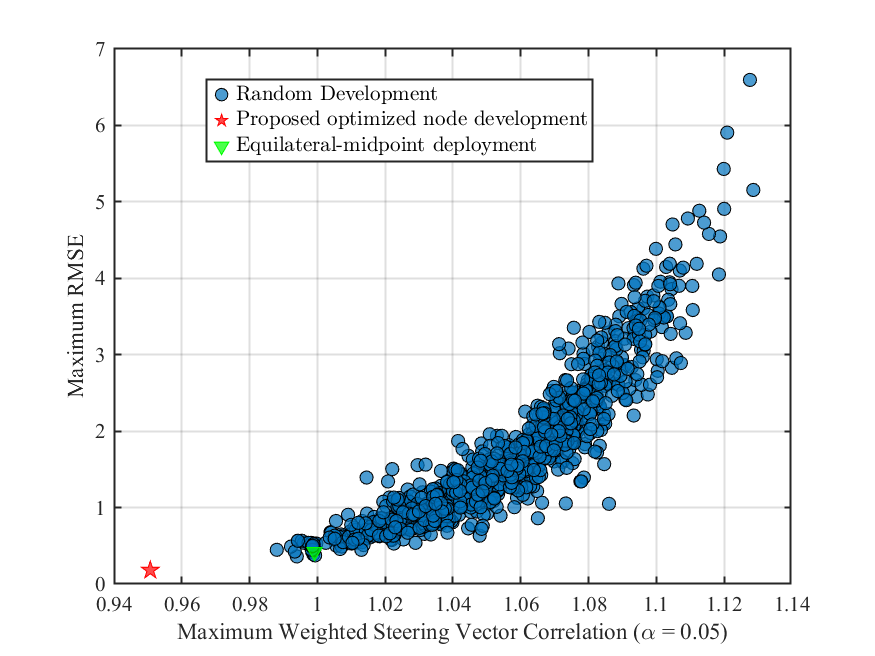}
        \caption{Comparison of optimized node Deployment against random deployments.}
        \label{r2}
    \end{minipage}
    \hfill
    \begin{minipage}[b]{0.32\linewidth}
        \centering
        \includegraphics[width=\linewidth]{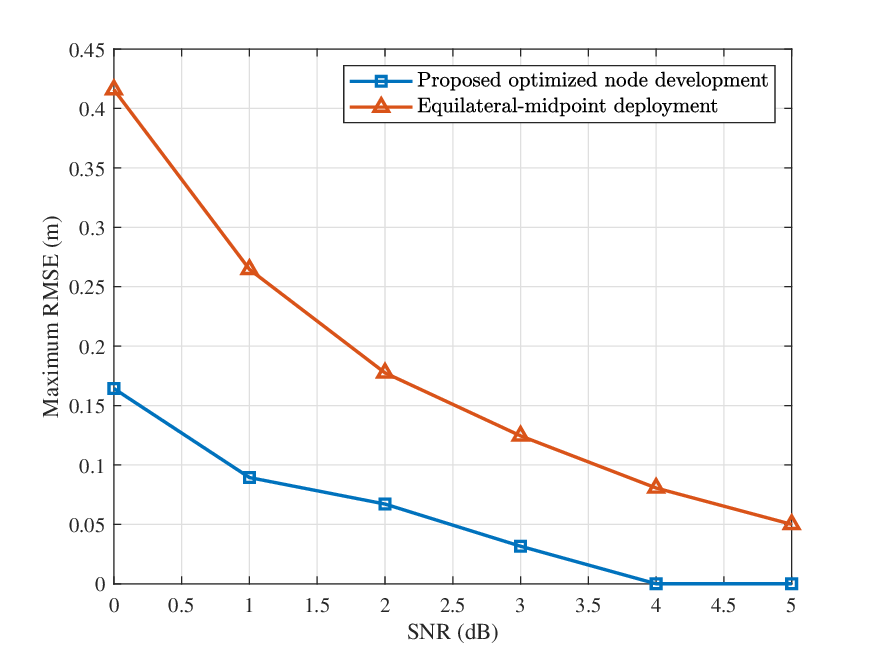}
        \caption{Impact of SNR on the maximum RMSE performance.}
        \label{r3}
    \end{minipage}
\end{figure*}
\begin{figure*}[!t]
    \centering
           \begin{minipage}[b]{0.32\linewidth}
        \centering
        \includegraphics[width=\linewidth]{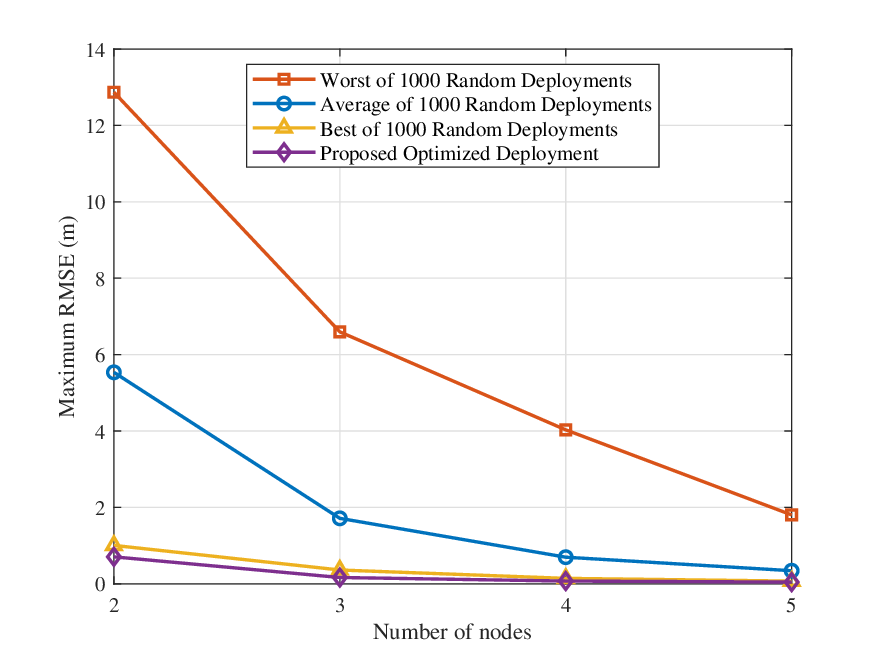}
        \caption{Impact of the number of nodes on the maximum RMSE performance.}
        \label{r4}
    \end{minipage}
    \hfill
    \begin{minipage}[b]{0.32\linewidth}
        \centering
\includegraphics[width=\linewidth]{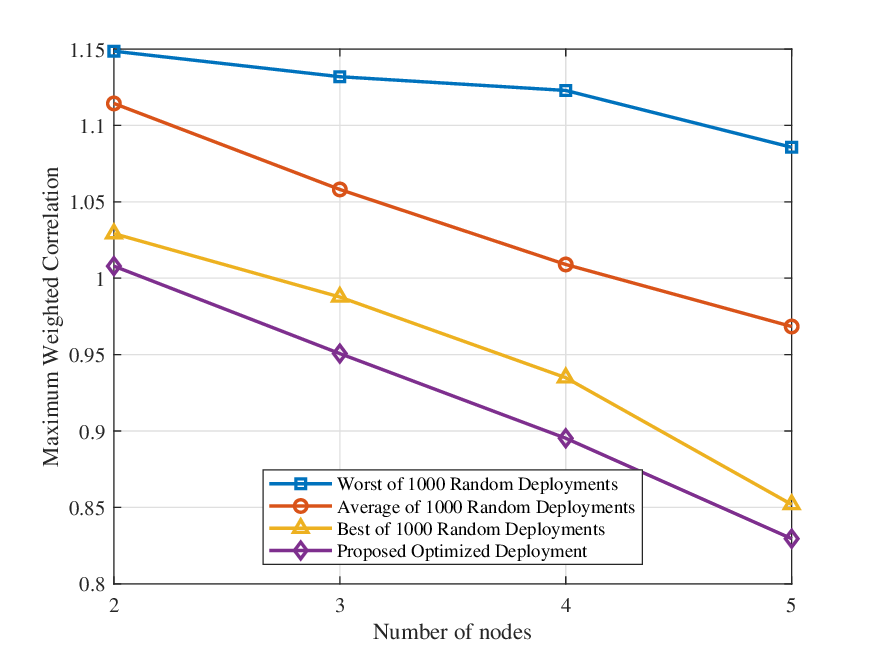}
    \caption{Impact of the number of nodes on the maximum weighted correlation performance.}
    \label{r5}
    \end{minipage}
    \hfill
    \begin{minipage}[b]{0.32\linewidth}
        \centering
\includegraphics[width=\linewidth]{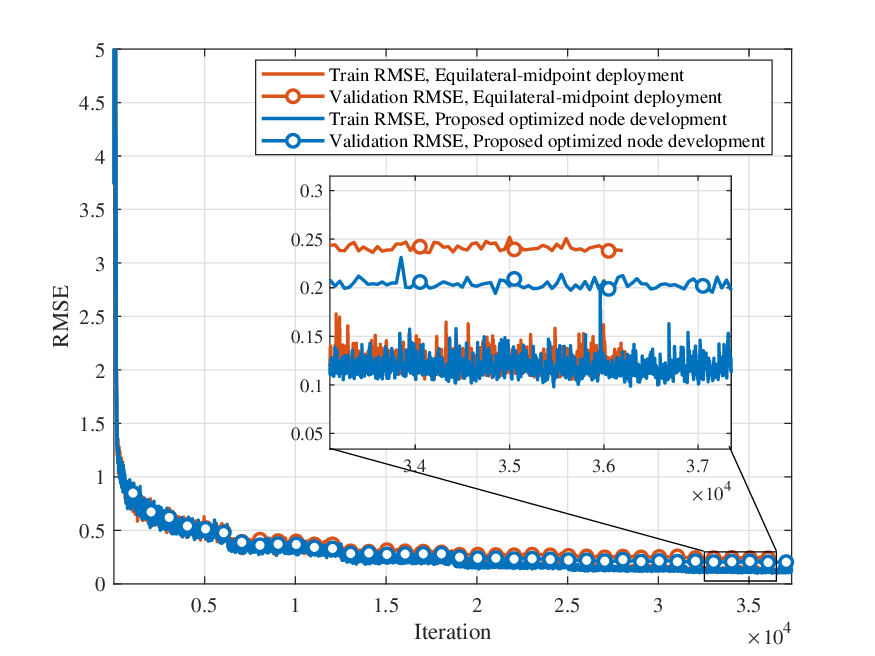}
        \caption{Training/validation RMSE of the regression network.}
        \label{r7}
    \end{minipage}
\end{figure*}
We adopt two representative approaches for performance evaluation: the two-dimensional MUSIC algorithm, which is widely used in subspace-based localization, and a NN-based method, which represents the emerging data-driven paradigm. 
For NN-based localization, the input feature is constructed by vectorizing the real and imaginary parts of the sample covariance matrix of the received snapshots. Two network architectures are adopted:
classification network and regression network.\footnote{Based on (1), we synthesize a dataset by randomly sampling target positions within $\mathcal{P}$ and generating received snapshots under randomized additive noise. A total of $100{,}000$ samples are created; $80\%$ are used for training and $20\%$ for validation. Implementation details for the NN-based localization can be found at \url{https://gitcode.com/qkq10/NN-based_Localization}.}
. 


\subsection{Convergence Verification}

Fig.~\ref{r1} validates the convergence behavior of the proposed GA-based optimization, $J=3$. It can be observed that the proposed deployment optimization consistently drives the maximum weighted correlation downward and stabilizes after about 400 generations.

\subsection{Robustness Verification with MUSIC-based Localization}
We first evaluate robustness under the MUSIC algorithm using $J=3$ cooperative nodes, where the maximum RMSE within the network is adopted as the metric. As shown in Fig.~\ref{r2}, three schemes are compared: the proposed optimized deployment, the fixed equilateral-midpoint deployment~\cite{10960698}, and $1000$ random deployments. Consistent with our conjecture, the optimized deployment achieves the lowest maximum weighted steering vector correlation, thereby yielding the smallest maximum RMSE. This validates the effectiveness of the proposed correlation-based optimization framework in enhancing worst-case localization performance.

\subsubsection{Impact of SNR on Robustness}
We next examine the impact of SNR in Fig.~\ref{r3}. The proposed deployment shows a clear advantage in the low-SNR regime, where errors are more likely due to noise-induced misidentification. As SNR increases, position-dependent signals become easier to distinguish, reducing the robustness challenge and narrowing the performance gap. Nevertheless, the proposed framework remains particularly beneficial for improving robustness in low-SNR scenarios.

\subsubsection{Impact of the Node Number on Robustness}
To further investigate the influence of cooperative node density, we simulate $2000$ random deployments and record the best, worst, and average performance. Fig.~\ref{r4} shows that the maximum RMSE statistically decreases with more nodes, implying improved robustness on average. However, the best two-node deployment can still outperform the worst five-node deployment, highlighting that robustness is highly deployment-dependent. In contrast, our optimized deployment consistently outperforms the best random deployment across all node numbers. Similar observations are obtained in Fig.~\ref{r5} using the maximum weighted steering vector correlation as the metric, further supporting the validity of the proposed optimization criterion.

\subsection{Robustness Verification with NN-based Localization}
We further validate robustness under NN-based localization. Fig.~\ref{r7} illustrates the training and validation curves for a regression network. Although the optimization is formulated in terms of worst-case robustness, the networks trained with optimized deployments achieve better validation performance compared with the equilateral-midpoint baseline. This is because reducing steering-vector correlation enhances the separability of covariance features, which indirectly benefits average-case learning.

\begin{table}[!t]
\centering
\caption{Robust performance comparison with NN-based localization.}
\begin{tabular}{lcc}
\toprule
\textbf{Deployment Strategy} & \textbf{Worst Acc.} & \textbf{Max RMSE} \\
\midrule
Proposed scheme & \textbf{0.9660} & \textbf{0.2824} \\
Equilateral-midpoint & 0.8632 & 0.4637 \\
\bottomrule
\end{tabular}
\label{tab:robust_performance}
\end{table}
Table~\ref{tab:robust_performance} summarizes the test results. In the classification task, the worst-case accuracy improves from $0.8632$ to \textbf{0.9660}. In the regression task, the maximum RMSE decreases from $0.4637$ to \textbf{0.2824}. These results confirm that robustness-oriented node deployment enhances NN-based localization across both classification and regression networks.



\section{CONCLUSION}
In this paper, we have explored the impact of node deployment on localization robustness in multi-node cooperative systems, specifically focusing on ensuring spatially continuous and resilient positioning services. By analyzing the relationship between node geometry and localization accuracy, we introduced a novel distance-weighted steering vector correlation metric that accurately captures the influence of node deployment on localization performance. Based on this, we formulated a new optimization framework to minimize the maximum weighted steering vector correlation, aimed at enhancing the robustness of localization services.
Moreover, we developed a GA to solve the optimization problems. Simulation results
validate the effectiveness of the
proposed methods, demonstrating significant improvements in
robust localization performance. 


\section*{Acknowledgment} Haojin Li and Kaiqian Qu contributed equally to the work.

\bibliographystyle{IEEEtran} 
\bibliography{IEEEabrv,bib}
\end{document}